\documentclass{scrartcl}
\KOMAoptions{paper=letter}

\usepackage{group}
\usepackage{enumitem}
\usepackage{bm}

\theoremstyle{definition}

\theoremstyle{plain}

\newcommand{\fnash}{f^{\mathrm{EQ}}}

\title{Equilibrium Play Without Mutual Knowledge of Rationality}

\author{Florian Brandl\\University of Bonn%
\and
Felix Brandt\\
Technical University of Munich%
}

\begin{document}

\maketitle

\begin{abstract}
Equilibrium play in two-player zero-sum games is usually justified via epistemic assumptions, such as mutual knowledge of rationality and beliefs, that go far beyond the rationality of the players.
We propose a justification that dispenses with these assumptions.
To this end, we consider solution concepts that assign to every subgame of a given game a set of plausible actions for each player, and we impose two conditions.
Rationality requires that the plausible sets are supports of undominated strategies or, equivalently, that all plausible actions are best responses to a common belief about the opponent.
Inheritance requires that plausible actions remain plausible when implausible actions are discarded.
In two-player games, the two conditions generically characterize the solution concepts that consistently select supports of Nash equilibria. Zero-sum payoffs ensure that Nash equilibria are generically unique, so that such a selection exists and is unique.
Equilibrium play thus emerges from individual rationality and the mutual understanding that plausibility judgments persist when implausible actions are discarded.
\end{abstract}

\section{Introduction}

Zero-sum games occupy a special place in game theory.
\Citet{vNeu28a} proved that every finite two-player zero-sum game admits maximin strategies that are interchangeable and yield unique payoffs---nowhere is the case for equilibrium play more compelling than in this class of games.
Yet, the classic justification of maximin play is decision-theoretic rather than strategic: a maximin strategy guarantees the value against every conceivable behavior of the opponent.
Why a player who faces a possibly naive opponent should nevertheless randomize according to an equilibrium strategy is much less clear, and has been debated at least since \citeauthor{Ells56a}'s (\citeyear{Ells56a}) ``reluctant duelist'': exploiting a flawed opponent seems to require abandoning equilibrium play.
Epistemic analyses answer that equilibrium behavior is warranted when the players' rationality and their beliefs are mutually known \citep{AuBr95a}---assumptions that go far beyond rationality of the players themselves.

This paper offers a justification of equilibrium play in zero-sum games that dispenses with mutual knowledge of rationality.
We consider solution concepts that assign to every subgame of a given game---every pair of sets of available actions---a set of plausible actions for each player, and we impose two conditions.
Rationality requires the plausible sets to be supports of undominated strategies.
Since a dominated strategy cannot be a best response to any belief about the opponent, rationality concerns each player in isolation and involves no assumption about what a player knows about the other player.
Inheritance requires that the solution be inherited by any reduced game that is obtained by discarding actions outside the solution: plausible actions remain plausible when implausible actions are removed.
In the intended interpretation, each player is rational, and it is mutually understood that plausibility judgments persist when implausible actions are discarded.

Our main result shows that for generic zero-sum games, these two conditions characterize equilibrium play: a solution concept satisfies inheritance and rationality if and only if it assigns to every subgame the supports of its unique Nash equilibrium.
Equilibrium play, in the sense of the set of actions that are played with positive probability in equilibrium, thus emerges from minimal rationality and inheritance alone.
The result also delimits how far the exploitation of a non-equilibrium opponent can go: while the probabilities within the equilibrium support are left unrestricted, no action outside the support is ever plausible.

Neither rationality nor inheritance refers to the zero-sum structure of the game, and we derive our main result from a more general analysis.
\Cref{sec:twoplayer} deals with arbitrary two-player games.
Without any genericity assumption, the two conditions imply that the plausible actions are sandwiched between the supports of a Nash equilibrium and the sets of best responses to it (\Cref{lem:sandwich}).
Hence, if all Nash equilibria of all subgames are quasi-strict---a property that holds for generic payoff matrices---the two conditions characterize the solution concepts that consistently select supports of Nash equilibria.
Such selections need not be unique, reflecting the classical coordination problem, and they need not exist at all, even when the game itself has a unique Nash equilibrium; both issues disappear when the Nash equilibria of all subgames are unique.
\Cref{sec:symmetric} concerns symmetric games, for which it is natural to restrict attention to symmetric subgames, so that the one-sided reductions underlying the general analysis cannot be used.
The characterization nevertheless persists, based on the fact---derived from the index theory of Nash equilibria---that if every symmetric subgame has a unique Nash equilibrium, all these equilibria have supports of odd size.
\Cref{sec:zerosum} collects the consequences for zero-sum games: uniqueness of Nash equilibria is generic in this class, which yields our main result, and tournament games, whose Nash equilibria are always unique, inherit the symmetric characterization, implying that no strict refinement of the bipartisan set satisfies inheritance.

Our result is best contrasted with rationalizability \citep{Bern84a,Pear84a}.
Rationalizability rests on a notion of rationality that applies to each action in isolation: an action is justifiable if it is a best response to some belief about the opponent, and different actions may be justified via different beliefs.
Our rationality condition applies to sets of actions instead: a set of actions is the support of an undominated strategy if and only if all of its elements are best responses to one and the same belief (\Cref{lem:support-beliefs}, which is due to \citealp{Wein20a}).
All plausible actions thus have to be rationalized by a common belief.
Mathematically, this equivalence follows from the characterization of undominated mixed strategies as best responses to beliefs \citep[Lemma~3]{Pear84a} and the linearity of expected payoffs; conceptually, it replaces one belief per action by one belief for all plausible actions.

The interactive assumptions can be contrasted as well.
Pearce and Bernheim conclude from common knowledge of rationality that players choose rationalizable actions, which, in two-player games, are exactly the actions surviving the iterated elimination of dominated actions.
Inheritance can take over the role of common knowledge of rationality: every solution concept satisfying inheritance and action-wise rationality selects rationalizable actions only, and the rationalizable sets themselves constitute the largest such solution concept (\Cref{prop:rationalizable}).
Strengthening rationality from one belief per action to one belief for all plausible actions then sharpens the prediction from rationalizability all the way down to Nash equilibrium.
Neither rationalizability nor our characterization requires active randomization.
Our result identifies equilibrium supports without prescribing the probabilities of play.

Alternative approaches to characterizing equilibrium play are based on notions of consistency across games \citep[see, e.g.,][]{PeTi96a,NPRV96a,DNRT01a,BrBr17c,BrBr23a}. 
In contrast to these methods, our result has the advantage of being applicable to a single game and its subgames without needing to reference additional hypothetical games. 
Set-valued solution concepts for zero-sum games go back to \citet{Shap64a}; prominent examples for general games include sets that are closed under rational behavior \citep{BaWe91a}.
The support of equilibrium play, our $\fnash$, is known as the essential set for symmetric zero-sum games \citep{DuLa99a} and as the bipartisan set in tournament games \citep{LLL93b}.
Conditions in the spirit of inheritance have a long history in the theory of choice consistency: inheritance weakens Chernoff's contraction condition \citep{Cher54a}, while strengthening inheritance by the reverse inclusion yields the strong superset property \citep{Bord79a}, known as the outcast condition in choice theory \citep{AiMa81a}.

\section{The Model}

Let $m$ and $n$ be positive integers and $[m] = \{1,\dots,m\}$.
For $x,y\in\mathbb R^k$, $x \le y$ if $x_i \le y_i$ for all $i\in[k]$, and $x \ll y$ if $x_i < y_i$ for all $i\in[k]$.
Moreover, $\bm 1$ denotes the all-ones vector of appropriate dimension and $\supp(x) = \{i\in[k]\colon x_i\neq 0\}$ the support of $x\in\mathbb R^k$.
Let $\Delta(m) = \{p\in\mathbb R_+^m\colon \sum_{i\in [m]} p_i = 1\}$ be the set of strategies over $[m]$ and, for nonempty $S\subseteq [m]$, $\Delta(S) = \{p\in\Delta(m)\colon \supp(p)\subseteq S\}$ the set of strategies with support in $S$; $\Delta(n)$ and $\Delta(T)$ for nonempty $T\subseteq[n]$ are defined analogously.

A two-player game is a pair of matrices $A,B\in\mathbb R^{m\times n}$, where $A_{ij}$ and $B_{ij}$ are the payoffs of the row player and the column player for the action profile $(i,j)\in[m]\times[n]$; zero-sum games correspond to $B = -A$.
For nonempty $S\subseteq[m]$ and $T\subseteq[n]$, the subgame $(A,B)_{S,T} = (A_{S,T},B_{S,T})$ of $(A,B)$ is the game in which the row player's actions are restricted to $S$ and the column player's actions are restricted to $T$; in particular, $(A,B)_{[m],[n]} = (A,B)$.
We identify the strategies of the two players in $(A,B)_{S,T}$ with $\Delta(S)$ and $\Delta(T)$.
A Nash equilibrium of $(A,B)_{S,T}$ is then a pair $(p,q)\in\Delta(S)\times\Delta(T)$ such that ${p'}^tAq \le p^tAq$ and $p^tBq' \le p^tBq$ for all $p'\in\Delta(S)$ and $q'\in\Delta(T)$.
We call a two-player game $(A,B)$ \emph{generic} if every subgame is nondegenerate: in each subgame, every strategy of either player with support of size $k$ has at most $k$ pure best responses of the other player \citep[cf.\ Definition~2]{vSte21b}.
This condition is inherited by subgames and holds for almost all payoff pairs $(A,B)$, and also for almost all payoff matrices $A$ within the class of zero-sum games $(A,-A)$.

For a fixed game $(A,B)$, a solution concept is a function $f$ that assigns to each pair of nonempty action sets $S\subseteq[m]$ and $T\subseteq[n]$ a pair $f(S,T) = (S',T')$ of nonempty subsets $S'\subseteq S$ and $T'\subseteq T$.
For such pairs of sets, we write $(S',T')\subseteq (S,T)$ if $S'\subseteq S$ and $T'\subseteq T$.
We order solution concepts by pointwise inclusion: $f\subseteq g$ if $f(S,T)\subseteq g(S,T)$ for every subgame $(A,B)_{S,T}$.
If every subgame of $(A,B)$ has a unique Nash equilibrium, one solution concept is the \emph{Nash equilibrium support} $\fnash$, which returns the pair of supports of the unique Nash equilibrium of every subgame.

A solution concept $f$ for $(A,B)$ satisfies inheritance if for all nonempty $S\subseteq[m]$ and $T\subseteq[n]$ and all nonempty $S'\subseteq S$ and $T'\subseteq T$,
\begin{align}
    f(S,T) \subseteq (S',T') \text{ implies } f(S,T) \subseteq f(S',T'). \tag{inheritance}
\end{align}
That is, when actions outside the solution of a subgame are discarded, the reduced subgame inherits the solution.
Observe that if $f(S,T) = (S',T')$, then inheritance implies $f(S',T') = (S',T')$, since $f(S,T) \subseteq f(S',T') \subseteq (S',T') = f(S,T)$.
Dominance refers to a player's own payoffs: a strategy $p\in\Delta(S)$ of the row player is dominated in $(A,B)_{S,T}$ if some $\tilde p\in\Delta(S)$ satisfies $(p^tA)_j < (\tilde p^tA)_j$ for every $j\in T$; analogously, a strategy $q\in\Delta(T)$ of the column player is dominated if some $\tilde q\in\Delta(T)$ satisfies $(Bq)_i < (B\tilde q)_i$ for every $i\in S$.
A solution concept $f$ for $(A,B)$ satisfies rationality if for all nonempty $S\subseteq[m]$ and $T\subseteq[n]$,
\begin{align}
    \begin{aligned}
        &f(S,T) = (S',T') \text{ implies } S' = \supp(p) \text{ and } T' = \supp(q) \text{ for some}\\
    &\text{strategies } p\in\Delta(S) \text{ and } q \in\Delta(T) \text{ that are undominated in } (A,B)_{S,T}.
    \end{aligned}
    \tag{rationality}
\end{align}
Equivalently, every strategy supported on the selected actions must be undominated; this follows from \Cref{lem:support-beliefs} below.

\section{Two-Player Games}\label{sec:twoplayer}

Our results rely on a correspondence between dominance and beliefs, which substantiates the interpretation of rationality given in the introduction: the components of the solution are exactly the sets of actions all of which are best responses to a common belief.
Here, a strategy $p\in\Delta(S)$ is a best response to a belief $q\in\Delta(T)$ in $(A,B)_{S,T}$ if ${p'}^tAq \le p^tAq$ for all $p'\in\Delta(S)$, and an action $i\in S$ is a best response to $q$ in $(A,B)_{S,T}$ if $(Aq)_i = \max_{i'\in S}\, (Aq)_{i'}$.
Analogously, an action $j\in T$ of the column player is a best response to a belief $p\in\Delta(S)$ in $(A,B)_{S,T}$ if $(p^tB)_j = \max_{j'\in T}\, (p^tB)_{j'}$.

\begin{lemma}[\citealp{Wein20a}, Proposition 1]\label{lem:support-beliefs}
    Let $S\subseteq[m]$, $T\subseteq[n]$, and $S'\subseteq S$ be nonempty.
    Then, $S'$ is the support of a strategy that is undominated in $(A,B)_{S,T}$ if and only if there exists a belief $q\in\Delta(T)$ such that every action in $S'$ is a best response to $q$ in $(A,B)_{S,T}$.
    The analogous statement holds for the column player.
\end{lemma}
\begin{proof}
    By \citet[Lemma~3]{Pear84a}, a strategy $p\in\Delta(S)$ is undominated in $(A,B)_{S,T}$ if and only if it is a best response to some belief $q\in\Delta(T)$.
    By linearity of expected payoffs, $p$ is a best response to $q$ if and only if every action in $\supp(p)$ is a best response to $q$.
    The row-player statement follows by choosing $p$ with $\supp(p) = S'$, and the column-player statement follows analogously.
\end{proof}

\Cref{lem:support-beliefs} also shows that rationality has an equivalent universal formulation: if $f(S,T) = (S',T')$, then every strategy in $\Delta(S')$ and every strategy in $\Delta(T')$ is undominated in $(A,B)_{S,T}$.
Indeed, a belief rationalizing every action in $S'$ makes every strategy in $\Delta(S')$ a best response and hence undominated; the column-player argument is analogous.
Conversely, the universal formulation includes strategies with support exactly $S'$ and $T'$.
We use both formulations of rationality interchangeably.

Rationalizability is based on the action-wise version of rationality.
An action is undominated in $(A,B)_{S,T}$ if the strategy that plays it with probability one is undominated, and a solution concept $f$ for $(A,B)$ satisfies \emph{action-wise rationality} if, for every subgame $(A,B)_{S,T}$, all actions in both components of $f(S,T)$ are undominated in $(A,B)_{S,T}$.
By the universal formulation, rationality implies action-wise rationality.
For a subgame $(A,B)_{S,T}$, let $(S^0,T^0) = (S,T)$ and, for $k\ge 0$, let $(S^{k+1},T^{k+1})$ be the pair of sets of actions that are undominated in $(A,B)_{S^k,T^k}$.
This sequence is eventually constant; we denote its limit by $R(S,T)$, which consists of the rationalizable actions of $(A,B)_{S,T}$ \citep{Bern84a,Pear84a}.
The following proposition shows that inheritance can take over the role of common knowledge of rationality in the derivation of rationalizability.

\begin{proposition}\label{prop:rationalizable}
    Let $A,B\in\mathbb R^{m\times n}$.
    Then, $R$ is the largest solution concept for $(A,B)$ satisfying inheritance and action-wise rationality.
\end{proposition}
\begin{proof}
    Consider a subgame $(A,B)_{S,T}$, and let $(R_S,R_T) = R(S,T)$.
    For each action $a\in R_S$, \Cref{lem:support-beliefs} yields a belief $q\in\Delta(R_T)$ to which $a$ is a best response in $(A,B)_{R_S,R_T}$, since $a$ is undominated in this subgame.
    No action $i\in S\setminus R_S$ is a best response to $q$ in $(A,B)_{S,T}$: such an action is dominated in $(A,B)_{S^k,T^k}$ for some $k$ by some $\tilde p\in\Delta(S^k)$, and $q\in\Delta(T^k)$ implies $\tilde p^tAq > (Aq)_i$.
    Hence, $\max_{i\in S}\,(Aq)_i = \max_{i\in R_S}\,(Aq)_i = (Aq)_a$, so $a$ is a best response to $q$ in $(A,B)_{S,T}$.
    Analogously, every action in $R_T$ is a best response in $(A,B)_{S,T}$ to some belief in $\Delta(R_S)$.
    By \Cref{lem:support-beliefs}, all actions in $R_S$ and $R_T$ are therefore undominated in $(A,B)_{S,T}$, so $R$ satisfies action-wise rationality.

    For inheritance, let $(R_S,R_T)\subseteq (S',T')\subseteq(S,T)$, and let $(\hat S^k,\hat T^k)$ be the sequence defined above for the subgame $(A,B)_{S',T'}$.
    If $(R_S,R_T)\subseteq(\hat S^k,\hat T^k)$, then, since $\hat S^k\subseteq S$, every $a\in R_S$ is a best response in $(A,B)_{\hat S^k,\hat T^k}$ to a belief in $\Delta(R_T)\subseteq\Delta(\hat T^k)$; hence, $a$ is undominated in $(A,B)_{\hat S^k,\hat T^k}$ by \Cref{lem:support-beliefs}, i.e., $a\in\hat S^{k+1}$.
    Together with the analogous argument for $R_T$, induction yields $(R_S,R_T)\subseteq R(S',T')$.
    Thus, $R$ satisfies inheritance.

    Finally, let $f$ be any solution concept for $(A,B)$ satisfying inheritance and action-wise rationality.
    We show by induction that $f(S,T)\subseteq (S^k,T^k)$ for all $k$.
    The inclusion holds for $k = 0$ because $f(S,T)\subseteq(S,T)$.
    If it holds for $k$, then inheritance yields $f(S,T)\subseteq f(S^k,T^k)$, which is contained in $(S^{k+1},T^{k+1})$ by action-wise rationality.
    Hence, $f(S,T)\subseteq R(S,T)$ for every subgame $(A,B)_{S,T}$, so $R$ is the largest such solution concept.
\end{proof}

In general, inheritance and rationality cannot single out one solution concept, because Nash equilibria need not be unique even in generic games: generic two-player games have a finite odd number of Nash equilibria \citep{Hars73a}, and this number can be larger than one on open sets of games.

\begin{example}\label{ex:coordination}
    Consider the coordination game
    \begin{align*}
        (A,B) = \begin{pmatrix} (2,2) & (0,0)\\ (0,0) & (1,1) \end{pmatrix}
    \end{align*}
    with two pure Nash equilibria, in which both players choose action~1 or both choose action~2, as well as a mixed Nash equilibrium in which both players play $(\tfrac{1}{3},\tfrac{2}{3})$.
    On the proper subgames, rationality uniquely determines the solution because each player has a single undominated action; for the full game, each of the three pairs of Nash equilibrium supports---$(\{1\},\{1\})$, $(\{2\},\{2\})$, and $(\{1,2\},\{1,2\})$---yields a solution concept that satisfies inheritance and rationality.
\end{example}

What rationality and inheritance do imply---for arbitrary two-player games and without any genericity assumption---is the following.

\begin{lemma}\label{lem:sandwich}
    Let $A,B\in\mathbb R^{m\times n}$, and let $f$ be a solution concept for $(A,B)$ that satisfies inheritance and rationality.
    Then, for every subgame $(A,B)_{S,T}$ with $f(S,T) = (S',T')$, there is a Nash equilibrium $(p,q)$ of $(A,B)_{S,T}$ such that
    \begin{align*}
        \supp(p)\subseteq S'\subseteq \mathit{BR}(q) \qquad\text{and}\qquad \supp(q)\subseteq T'\subseteq \mathit{BR}(p),
    \end{align*}
    where $\mathit{BR}(q) = \{i\in S\colon (Aq)_i = \max_{i'\in S}\,(Aq)_{i'}\}$ and $\mathit{BR}(p) = \{j\in T\colon (p^tB)_j = \max_{j'\in T}\,(p^tB)_{j'}\}$ denote the sets of best responses to $q$ and $p$ in $(A,B)_{S,T}$.
\end{lemma}
\begin{proof}
    By inheritance, $(S',T')\subseteq f(S,T')$ and $(S',T')\subseteq f(S',T)$.
    By rationality and \Cref{lem:support-beliefs} applied to the subgame $(A,B)_{S,T'}$, there is a belief $q\in\Delta(T')$ to which every action in the row component of $f(S,T')$---in particular, every action in $S'$---is a best response, i.e., $S'\subseteq \mathit{BR}(q)$.
    Analogously, there is a belief $p\in\Delta(S')$ with $T'\subseteq \mathit{BR}(p)$.
    Then, $p$ is a best response to $q$ in $(A,B)_{S,T}$ because $\supp(p)\subseteq S'\subseteq \mathit{BR}(q)$, and $q$ is a best response to $p$ because $\supp(q)\subseteq T'\subseteq \mathit{BR}(p)$.
    Hence, $(p,q)$ is a Nash equilibrium of $(A,B)_{S,T}$ with the required inclusions.
\end{proof}

\Cref{lem:sandwich} states that all plausible actions are best responses to the conjectures of a single Nash equilibrium whose supports are themselves plausible.
Note that its proof uses only two instances of inheritance per subgame, namely those in which all implausible actions of exactly one player are discarded. \Citet[Corollaries~2 and~3]{Wein20a} gives related characterizations of equilibrium existence in terms of dominance.

A Nash equilibrium $(p,q)$ of $(A,B)_{S,T}$ is quasi-strict if $\mathit{BR}(q) = \supp(p)$ and $\mathit{BR}(p) = \supp(q)$.
Every Nash equilibrium of every subgame of a generic game is quasi-strict: the best-response condition and nondegeneracy imply
\begin{align*}
    |\supp(p)|\le |\mathit{BR}(q)|\le |\supp(q)|\le |\mathit{BR}(p)|\le |\supp(p)|,
\end{align*}
so equality holds throughout.

\begin{theorem}\label{thm:selection}
    Let $A,B\in\mathbb R^{m\times n}$ such that every Nash equilibrium of every subgame of $(A,B)$ is quasi-strict, and let $f$ be a solution concept for $(A,B)$ that satisfies inheritance.
    Then, $f$ satisfies rationality if and only if, for every subgame, $f(S,T)$ is the pair of supports of some Nash equilibrium of $(A,B)_{S,T}$.
\end{theorem}
\begin{proof}
    For the direction from left to right, \Cref{lem:sandwich} yields a Nash equilibrium $(p,q)$ of $(A,B)_{S,T}$ with $\supp(p)\subseteq S'\subseteq \mathit{BR}(q)$ and $\supp(q)\subseteq T'\subseteq \mathit{BR}(p)$, and quasi-strictness turns both chains of inclusions into equalities.
    For the direction from right to left, let $f(S,T) = (\supp(p),\supp(q))$ for some Nash equilibrium $(p,q)$ of $(A,B)_{S,T}$.
    Then, every action in $\supp(p)$ is a best response to the belief $q$ and every action in $\supp(q)$ is a best response to the belief $p$, so rationality follows from \Cref{lem:support-beliefs}.
\end{proof}

Under the assumptions of \Cref{thm:selection}, inheritance and rationality thus characterize the inheritance-consistent selections of Nash equilibrium supports rather than a single solution concept, and \Cref{ex:coordination} shows that this multiplicity cannot be avoided.
The remaining freedom is the classical coordination problem: the axioms discipline the plausible actions to be supported by the conjectures of some Nash equilibrium of every subgame, but which equilibrium the players settle on is a matter of convention.
This conclusion parallels the consistency-based characterization of the Nash equilibrium correspondence by \citet{NPRV96a}.
Moreover, inheritance-consistent selections need not exist at all: inheritance ties the solution of a subgame to the solutions of all larger subgames, and different larger subgames may pull toward different equilibria.

\begin{example}\label{ex:nonexistence}
    Consider the game
    \begin{align*}
        (A,B) = \begin{pmatrix}
            (3,3) & (0,0) & (2,2)\\
            (0,0) & (2,2) & (0,1)\\
            (1,3) & (3,0) & (1,1)\\
            (4,1) & (1,3) & (3,0)
        \end{pmatrix}.
    \end{align*}
    One can check that $(A,B)$ is generic, so that all Nash equilibria of all subgames are quasi-strict; the same holds for all sufficiently small perturbations of the payoffs.
    Suppose that $f$ is a solution concept for $(A,B)$ that satisfies inheritance and rationality.
Since rationality implies action-wise rationality, \Cref{prop:rationalizable} and iterated dominance give that $f(\{1,2,3\},\{1,2,3\}) = (\{1\},\{1\})$.
    Likewise, $f(\{1,2,4\},\{1,2,3\}) = (\{2\},\{2\})$.
    Both solutions are contained in $(\{1,2\},\{1,2,3\})$, so inheritance requires $(\{1,2\},\{1,2\})\subseteq f(\{1,2\},\{1,2,3\})$.
    Yet, against rows~1 and~2, the column strategy $q = (\tfrac35,\tfrac25,0)$ yields payoffs $(\tfrac95,\tfrac45)$, whereas pure column~3 yields $(2,1)$.
    Thus, $q$ is strictly dominated, although its support is contained in the column component of the solution, contradicting the universal formulation of rationality.
    Hence, no solution concept for $(A,B)$ satisfies inheritance and rationality. All dominance inequalities used above are strict, so the same holds for all sufficiently small perturbations of the payoffs.
    Note that $(A,B)$ itself has a unique Nash equilibrium, with supports $(\{3,4\},\{1,2\})$.
\end{example}

If every subgame has a unique Nash equilibrium, however, both problems disappear---a selection exists and is unique---and the characterization takes the following form.

\begin{corollary}\label{cor:unique}
    Let $A,B\in\mathbb R^{m\times n}$ such that every subgame of $(A,B)$ has a unique Nash equilibrium, and let $f$ be a solution concept for $(A,B)$.
    Then, $f$ satisfies inheritance and rationality if and only if $f = \fnash$.
\end{corollary}
\begin{proof}
    Since every two-player game has a quasi-strict Nash equilibrium \citep{Nord99a}, the unique Nash equilibrium of every subgame of $(A,B)$ is quasi-strict.
    If $f$ satisfies inheritance and rationality, then, by \Cref{thm:selection}, $f(S,T)$ is the pair of supports of the unique Nash equilibrium of $(A,B)_{S,T}$ for every subgame, i.e., $f = \fnash$.
    Conversely, $\fnash$ satisfies inheritance because a Nash equilibrium of $(A,B)_{S,T}$ whose supports are contained in $(S',T')$ is also a Nash equilibrium of $(A,B)_{S',T'}$ and hence the unique one; rationality then follows from \Cref{lem:support-beliefs}.
\end{proof}

\Cref{cor:unique} covers, for example, all games in which every subgame is solvable by the iterated elimination of strictly dominated actions, such as the prisoner's dilemma.

The analysis remains intrinsically two-player, however.
In a two-player game, a belief about the opponent is itself a mixed strategy of the opponent, so a pair of mutually justifying beliefs has exactly the form of a Nash equilibrium.
With three or more players, a player's belief is instead a distribution over profiles of opponents' actions.
Such beliefs may correlate the opponents' actions, and different players may induce different marginal beliefs about the same third player.
Thus, a direct many-player extension would require additional assumptions ensuring both agreement and independence of conjectures.
This is the same two-player/many-player divide identified by \citet{AuBr95a}: in two-player games, mutual knowledge of rationality, payoffs, and conjectures suffices, whereas their many-player result uses a common prior together with common knowledge of conjectures to recover a Nash equilibrium.

\section{Symmetric Games}\label{sec:symmetric}

A two-player game is symmetric if $n = m$ and $B = A^t$, so that the game is invariant under exchanging the roles of the players.
For symmetric games, it is natural to restrict attention to symmetric subgames, in which both players have the same set of actions; the one-sided reductions used in the proof of \Cref{lem:sandwich} are then unavailable.
The symmetric subgames of $(A,A^t)$ are the subgames $(A,A^t)_{S,S}$ for nonempty $S\subseteq[m]$, and a symmetric solution concept for the symmetric subgames of $(A,A^t)$ is a function $f$ with $f(S,S) = (S',S')$ for some nonempty $S'\subseteq S$, where inheritance and rationality are imposed for symmetric subgames only.
In symmetric games, both players have the same best responses to any given belief $z$, since $(z^tB)_j = (z^tA^t)_j = (Az)_j$; we write $\mathit{BR}_E(z) = \{i\in E\colon (Az)_i = \max_{i'\in E}\,(Az)_{i'}\}$ for this common set of best responses in $(A,A^t)_{E,E}$.
Moreover, if $(p,q)$ is a Nash equilibrium of a symmetric subgame, then so is $(q,p)$; hence, a unique Nash equilibrium of $(A,A^t)_{S,S}$ is of the form $(\sigma_S,\sigma_S)$, and we refer to $\supp(\sigma_S)$ as its support.
The following lemma, whose proof relies on the index theory of Nash equilibria, is the key structural fact for the symmetric case.

\begin{lemma}\label{lem:odd-support}
    Let $(A,A^t)$ be a symmetric two-player game such that every symmetric subgame has a unique Nash equilibrium.
    Then, the support of the unique Nash equilibrium of every symmetric subgame is of odd size.
\end{lemma}
\begin{proof}
    Consider a symmetric subgame $(A,A^t)_{S,S}$ with unique Nash equilibrium $(\sigma_S,\sigma_S)$, and let $E = \supp(\sigma_S)$.
    The restriction of $\sigma_S$ to $E$ is a completely mixed Nash equilibrium of $(A,A^t)_{E,E}$, which by assumption is the unique Nash equilibrium of $(A,A^t)_{E,E}$.
    It therefore suffices to show that a symmetric $k\times k$ game whose unique Nash equilibrium $(x,x)$ is completely mixed has $k$ odd; we write $A$ for the $k\times k$ payoff matrix, which we may assume to be positive, since adding a constant to all payoffs does not change the Nash equilibria.

    Since $x$ is completely mixed, all actions are best responses to $x$, so $Ax = v\bm 1$ for some $v\in\mathbb R$.
    We first show that $(x,x)$ is regular, in the sense that the matrix $\left(\begin{smallmatrix} A & -\bm 1 \\ \bm 1^t & 0\end{smallmatrix}\right)$ is nonsingular.
    Otherwise, there is $(d,\alpha)\neq 0$ with $Ad = \alpha\bm 1$ and $\bm 1^t d = 0$, where $d\neq 0$ because $d = 0$ would force $\alpha = 0$.
    Then, $x + \epsilon d$ is a completely mixed strategy for small enough $\epsilon > 0$, and $A(x+\epsilon d) = (v + \epsilon\alpha)\bm 1$, so $(x + \epsilon d, x + \epsilon d)$ is a second completely mixed Nash equilibrium, contradicting uniqueness.

    To conclude, we use the index of Nash equilibria, which was introduced by \citet{Shap74a}, in the determinant form of \citet[Section~6]{vSte21b}.
    There, a nondegenerate $k\times k$ game $(A',B')$ with positive payoffs is represented by the matrix $C' = \left(\begin{smallmatrix} 0 & A' \\ {B'}^t & 0\end{smallmatrix}\right)$; the index of a completely mixed Nash equilibrium is $(-1)^{2k+1}\operatorname{sign}\det C'$, where $\det C'\neq 0$ by nondegeneracy, and the indices of all Nash equilibria sum to $+1$.
    Since $(A,A^t)$ may be degenerate, we apply this to a nearby game.
    The completely mixed Nash equilibria $(p,q)$ of a game $(A',B')$ are the solutions of $A'q = u\bm 1$, ${B'}^tp = w\bm 1$, and $\bm 1^tp = \bm 1^tq = 1$ with $p,q\gg 0$.
    These equations form two linear systems whose coefficient matrices at $(A',B') = (A,A^t)$ both equal the nonsingular bordered matrix above.
    For sufficiently small payoff perturbations, both matrices remain nonsingular and their inverses vary continuously, so the systems have unique solutions $(p,q,u,w)$ with $(p,q)$ close to $(x,x)$ and $p,q\gg 0$.
    Thus, every nearby game has exactly one completely mixed Nash equilibrium near $(x,x)$.
    By compactness of the strategy spaces and the closed graph of the Nash equilibrium correspondence, every equilibrium of a sufficiently nearby game is close to the unique equilibrium $(x,x)$ and hence completely mixed.
    Now choose a positive matrix $A'$ close to $A$ such that $(A',{A'}^t)$ is nondegenerate, which is true for almost all $A'$.
    Then, $(A',{A'}^t)$ has a unique Nash equilibrium, which is completely mixed and, since the indices sum to $+1$, has index $+1$.
    On the other hand, $C' = \left(\begin{smallmatrix} 0 & A' \\ A' & 0\end{smallmatrix}\right)$ satisfies $\det C' = (-1)^k(\det A')^2$, so this index equals $(-1)^{2k+1}(-1)^k = (-1)^{k+1}$.
    Hence, $(-1)^{k+1} = +1$, so $k$ is odd.\footnote{\label{fn:skew}For symmetric zero-sum games, the index theory can be avoided. Without the normalization to positive payoffs, $A$ is then skew-symmetric, and so is the matrix $\left(\begin{smallmatrix} A & -\bm 1 \\ \bm 1^t & 0\end{smallmatrix}\right)$, which is nonsingular by the regularity of $(x,x)$. Since skew-symmetric matrices of odd order are singular, $k+1$ is even, so $k$ is odd. For completely mixed symmetric zero-sum games, this observation goes back to \citet[p.~479]{Kapl45a}.}
\end{proof}

\begin{theorem}\label{thm:symmetric}
    Let $(A,A^t)$ be a symmetric two-player game such that every symmetric subgame of $(A,A^t)$ has a unique Nash equilibrium, and let $f$ be a symmetric solution concept for the symmetric subgames of $(A,A^t)$, where $\fnash$ returns the pair of supports of the unique Nash equilibrium of every symmetric subgame.
    Then, $f$ satisfies inheritance and rationality if and only if $f = \fnash$.
\end{theorem}
\begin{proof}
    That $\fnash$ satisfies inheritance and rationality follows as in the proof of \Cref{cor:unique}, with all subgames replaced by symmetric subgames.

    For the converse, let $f$ satisfy inheritance and rationality, and consider a symmetric subgame $(A,A^t)_{S,S}$ with $f(S,S) = (S',S')$.
    By inheritance, $f(S',S') = (S',S')$, and rationality together with \Cref{lem:support-beliefs} yields a belief $q\in\Delta(S')$ such that $S'\subseteq \mathit{BR}_{S'}(q)$.
    Then, $(q,q)$ is a Nash equilibrium of $(A,A^t)_{S',S'}$ and hence coincides with the unique Nash equilibrium $(\sigma_{S'},\sigma_{S'})$.
    Since unique Nash equilibria are quasi-strict \citep{Nord99a}, $S'\subseteq \mathit{BR}_{S'}(\sigma_{S'}) = \supp(\sigma_{S'})\subseteq S'$.
    Thus, $\supp(\sigma_{S'}) = S'$, the size of $S'$ is odd by \Cref{lem:odd-support}, and $(A\sigma_{S'})_s = v$ for all $s\in S'$, where $v = \max_{s\in S'}\,(A\sigma_{S'})_s$.

    Next, let $i\in S\setminus S'$ and $E = S'\cup\{i\}$.
    By inheritance, $(S',S')\subseteq f(E,E)$, and rationality together with \Cref{lem:support-beliefs} yields a belief $q'\in\Delta(E)$ with $S'\subseteq \mathit{BR}_E(q')$.
    We claim that $(A\sigma_{S'})_i < v$.
    Suppose otherwise and consider $z_\lambda = (1-\lambda)\,\sigma_{S'} + \lambda\, q'$ for $\lambda\in[0,1]$.
    At both endpoints, $(Az_\lambda)_s$ is the same for all $s\in S'$---for $z_0 = \sigma_{S'}$ because $(A\sigma_{S'})_s = v$ for all $s\in S'$, and for $z_1 = q'$ because $S'\subseteq \mathit{BR}_E(q')$---and hence for every $\lambda$; denote this common value by $u_\lambda$.
    Now, $(Az_0)_i \ge u_0$ by supposition and $(Az_1)_i \le u_1$ because $S'\subseteq\mathit{BR}_E(q')$, so, by continuity, $(Az_{\lambda^*})_i = u_{\lambda^*}$ for some $\lambda^*\in[0,1]$.
    Then, every action in $E$ is a best response to $z_{\lambda^*}$, so $(z_{\lambda^*},z_{\lambda^*})$ is a Nash equilibrium of $(A,A^t)_{E,E}$ and hence coincides with $(\sigma_E,\sigma_E)$; quasi-strictness implies $\supp(\sigma_E) = \mathit{BR}_E(\sigma_E) = E$.
    This contradicts \Cref{lem:odd-support}, because $|E| = |S'| + 1$ is even.

    Hence, $(A\sigma_{S'})_i < v$ for all $i\in S\setminus S'$, so $\mathit{BR}_S(\sigma_{S'}) = S'$ and $(\sigma_{S'},\sigma_{S'})$ is a Nash equilibrium of $(A,A^t)_{S,S}$.
    It therefore coincides with $(\sigma_S,\sigma_S)$, and $f(S,S) = (S',S') = (\supp(\sigma_S),\supp(\sigma_S)) = \fnash(S,S)$.
\end{proof}

In the next section, we apply this result to symmetric zero-sum games and, in particular, to tournament games, whose Nash equilibria are always unique.

\section{Zero-Sum Games}\label{sec:zerosum}

A zero-sum game is a two-player game with $B = -A$; we identify it with the row player's payoff matrix $A\in\mathbb R^{m\times n}$, so that $A_{ij}$ is the row player's payoff and $-A_{ij}$ the column player's payoff for the action profile $(i,j)$.
The column player's best responses to a belief $p\in\Delta(S)$ in $A_{S,T}$ are then the actions $j\in T$ that minimize $(p^tA)_j$, and a strategy $q\in\Delta(T)$ of the column player is dominated in $A_{S,T}$ if some $\tilde q\in\Delta(T)$ satisfies $(A\tilde q)_i < (Aq)_i$ for every $i\in S$.
Every generic zero-sum game has a unique Nash equilibrium \citep[cf.][]{BKS50a}.
Indeed, a nondegenerate two-player game has only finitely many Nash equilibria \citep{vSte21b}, whereas two distinct equilibria of a zero-sum game would yield a continuum by convexity of the players' sets of optimal strategies.
Since genericity is inherited by subgames, every subgame of a generic zero-sum game also has a unique Nash equilibrium.
Our main result for zero-sum games is now an immediate consequence of \Cref{cor:unique}.

\begin{corollary}\label{cor:zerosum}
    Let $A\in\mathbb R^{m\times n}$ be a generic zero-sum game, and let $f$ be a solution concept for $A$.
    Then, $f$ satisfies inheritance and rationality if and only if $f = \fnash$.
\end{corollary}
\begin{proof}
    Every subgame of a generic zero-sum game has a unique Nash equilibrium, as shown above, so the statement follows from \Cref{cor:unique}.\footnote{For zero-sum games, the existence of quasi-strict Nash equilibria used in the proof of \Cref{cor:unique} goes back to \citet{BKS50a}.}
\end{proof}

In contrast to \Cref{ex:nonexistence}, a solution concept satisfying inheritance and rationality thus always exists for generic zero-sum games.

\begin{remark}[Independence of the axioms]\label{rem:independence}
    Rationality and inheritance are independent.
    The solution concept that returns, in every subgame, the actions that maximize a player's total payoff over the opponent's available actions---equivalently, the best responses to the uniform belief about the opponent---satisfies rationality by \Cref{lem:support-beliefs}, but violates inheritance.
    For example, for the generic game
    \begin{align*}
        A = \begin{pmatrix} 1 & \hphantom{-}2\\ 4 & -2 \end{pmatrix},
    \end{align*}
    it returns $(\{1\},\{2\})$, since the first row maximizes $\sum_{j} A_{ij}$ and the second column minimizes $\sum_{i} A_{ij}$; but for the subgame $A_{\{1\},\{1,2\}}$, it returns $(\{1\},\{1\})$, because the column player's best response to the first row is the first column.
    Hence, $(\{1\},\{2\})\subseteq (\{1\},\{1,2\})$, but $(\{1\},\{2\})\not\subseteq (\{1\},\{1\})$.
    Conversely, the trivial solution concept, which returns $(S,T)$ for every subgame $A_{S,T}$, satisfies inheritance, because $f(S,T)\subseteq (S',T')$ implies $(S',T') = (S,T)$, but it violates rationality for every generic game in which one of the players has more than one action: if the column player has at least two available actions, then any one-row, two-column subgame has distinct payoffs by genericity, and every full-support column strategy is dominated by the pure action with the lower payoff for the row player; the case of two row actions in a one-column subgame is analogous.
\end{remark}

A zero-sum game is symmetric as a two-player game if and only if $A = -A^t$, i.e., $A$ is skew-symmetric: then $B = -A = A^t$, so the game is both zero-sum and symmetric.
\Cref{thm:symmetric} therefore applies to symmetric zero-sum games; for these games, the index theory in the proof of \Cref{lem:odd-support} can be avoided (see \Cref{fn:skew}).

\begin{corollary}\label{cor:symzero}
    Let $A\in\mathbb R^{m\times m}$ be a symmetric zero-sum game such that every symmetric subgame of $A$ has a unique Nash equilibrium, and let $f$ be a symmetric solution concept for the symmetric subgames of $A$.
    Then, $f$ satisfies inheritance and rationality if and only if $f = \fnash$.
\end{corollary}
\begin{proof}
    Since $A$ is skew-symmetric, the game is symmetric as a two-player game, so the statement is a special case of \Cref{thm:symmetric}.
\end{proof}

By a theorem of \citet{LLL97a}, every symmetric zero-sum game whose off-diagonal entries are odd integers has a unique Nash equilibrium.
Since this property is inherited by symmetric subgames, all such games satisfy the assumption of \Cref{cor:symzero}.
The most prominent examples are \emph{tournament games}, whose off-diagonal entries are $1$ or $-1$, and \Cref{cor:symzero} then has an interesting consequence for tournament solutions.
A tournament is a complete and asymmetric relation $\succ$ on a finite set $X$ of alternatives; its tournament game is the symmetric zero-sum game $A\in\mathbb R^{X\times X}$ with $A_{ab} = 1$ if $a\succ b$, $A_{ab} = -1$ if $b\succ a$, and $A_{aa} = 0$.
The support of the unique Nash equilibrium $(\sigma_X,\sigma_X)$ of a tournament game is known as the bipartisan set $\bp(X)$ \citep{LLL93b}.
A \emph{tournament solution} $F$ maps every tournament to a nonempty subset of its alternatives; for $S\subseteq X$, we write $F(S)$ for the choice of $F$ from the restriction of $\succ$ to $S$.
$F$ \emph{refines} $\bp$ if $F(S)\subseteq\bp(S)$ for every tournament on every set $S$, and $F$ satisfies \emph{inheritance} (known as $\widehat\alpha_\supseteq$ in this context) if $F(X)\subseteq S\subseteq X$ implies $F(X)\subseteq F(S)$.
\citet[Corollary~2]{BBSS14a} have shown that no strict refinement of $\bp$ satisfies stability, a property that is significantly stronger than inheritance.
The following corollary strengthens this result.

\begin{corollary}\label{cor:tournaments}
    No strict refinement of $\bp$ satisfies inheritance.
\end{corollary}
\begin{proof}
    Let $F$ be a tournament solution that refines $\bp$ and satisfies inheritance, and consider the tournament game $A$ of an arbitrary tournament on a finite set $X$.
    Then, $f$ with $f(S,S) = (F(S),F(S))$ for all nonempty $S\subseteq X$ is a symmetric solution concept for the symmetric subgames of $A$, and inheritance for $F$ is precisely inheritance for $f$.
    Moreover, $f$ satisfies rationality: because $(\sigma_S,\sigma_S)$ is a Nash equilibrium of $A_{S,S}$ with support $\bp(S)$, every action in $\bp(S)$---and hence every action in $F(S)\subseteq \bp(S)$---is a best response to the belief $\sigma_S$ in $A_{S,S}$, so $F(S)$ is the support of an undominated strategy by \Cref{lem:support-beliefs}.
    Hence, \Cref{cor:symzero} implies $f = \fnash$, and, in particular, $F(X) = \bp(X)$.
    Since the tournament was arbitrary, $F = \bp$.
\end{proof}

\section*{Acknowledgments}
This material is based on work supported by the Deutsche Forschungsgemeinschaft under grants {BR~2312/14-1} and {EXC-2047}. This project was inspired by unpublished work on impartial ordinalism and the essential set by Klaus Nehring. 
The authors thank Chris Dong and Patrick Lederer for helpful discussions, and Gabriel Ziegler for noting that \Cref{lem:support-beliefs} is due to \citet{Wein20a}.

\end{document}